\documentclass[11pt,reqno]{amsart}
\usepackage{amssymb}
\usepackage{amsmath}
\usepackage{amscd}
\begin{document}
\setlength{\oddsidemargin}{0cm} \setlength{\evensidemargin}{0cm}

\theoremstyle{plain}
\newtheorem{theorem}{Theorem}[section]
\newtheorem{prop}[theorem]{Proposition}
\newtheorem{lemma}[theorem]{Lemma}
\newtheorem{coro}[theorem]{Corollary}
\newtheorem{conj}[theorem]{Conjecture}

\theoremstyle{definition}
\newtheorem{defi}[theorem]{Definition}
\newtheorem{exam}[theorem]{Example}
\newtheorem{remark}[theorem]{Remark}

\numberwithin{equation}{section}

\title{Classification of $(n+3)$-dimensional metric $n$-Lie algebras}

\author{Qiaozhi Geng}\address{School of Mathematical Sciences and LPMC, Nankai University,
Tianjin 300071, P.R. China}
\author{Mingming Ren}\address{School of Mathematical Sciences and LPMC, Nankai University,
Tianjin 300071, P.R. China}
\author{Zhiqi Chen$^*$}
\address{School of Mathematical Sciences and LPMC, Nankai University,
Tianjin 300071, P.R. China}\email{chenzhiqi@nankai.edu.cn}
\thanks{$^*$ Corresponding author}

\def\shorttitle{metric $n$-Lie algebras}

\subjclass[2000]{17B60, 17D25.}


\begin{abstract} In this paper, we focus on $(n+3)$-dimensional metric $n$-Lie algebras. To begin with,
we give some properties on $(n+3)$-dimensional $n$-Lie algebras.
Then based on the properties, we obtain the classification of
$(n+3)$-dimensional metric $n$-Lie algebras.
\end{abstract}

\maketitle

\baselineskip=20pt

\section{Introduction}
The theory of $n$-Lie algebras is strongly connected with many other
fields, such as dynamics, geometries and string theory. Takhtajan
\cite{Ta} developed the geometrical ideas of Nambu mechanics
\cite{Na} and introduced an analogue of the Jacobi identity, which
connects the generalized Nambu mechanics with the theory of $n$-Lie
algebras introduced by Filippov \cite{Fi}. It is given in \cite{HHM}
some new 3-Lie algebras and applications in membrane, including the
Basu-Harvey equation and the Bagger-Lambert model. More applications
can be found in \cite{BWXA,Ga,MVV,Nak,Po1,Po2,VV}.

A class of $n$-Lie algebras (finite dimensional, real, $n>2$) which
have appeared naturally in mathematical physics are those which
possess a non-degenerate inner product which is invariant under the
inner derivations, which are called metric $n$-Lie algebras. They
have arisen for the first time in the work of Figueroa-O'Farrill and
Papadopoulos \cite{FP} in the classification of maximally
supersymmetric type $IIB$ supergravity backgrounds \cite{FP1}, and
more recently, for the case of $n=3$, in the work of Bagger and
Lambert \cite{BL,BL1} and Gustavsson \cite{Gu} on a superconformal
field theory for multiple M2-branes. It is this latter work which
has revived the interest of part of the mathematical physics
community on metric $n$-Lie algebras. There are some progress on
metric $n$-Lie algebras, such as the classification for Euclidean
\cite{Nag} (see also \cite{Pa}) and lorentzian metric $n$-Lie
algebras \cite{Fo}, the classification of index-2 metric 3-Lie
algebras \cite{MFM} and a structure theorem for metric $n$-Lie
algebras \cite{J}.

The classification of $n$-Lie algebras is more difficult than that
of Lie algebras due to the $n$-ary multilinear operation. As we
know, the only attempt to the classification of $(n+2)$-dimensional
$n$-Lie algebras is due to 6-dimensional 4-Lie algebras \cite{BS}.
The main goal of this paper is to classify $(n+3)$-dimensional
metric $n$-Lie algebras. To begin with, we get some properties on
$(n+3)$-dimensional $n$-Lie algebras by the basic results about the
$n$-Lie modules of simple $n$-Lie algebras. Furthermore, we give the
classification of $(n+3)$-dimensional metric $n$-Lie algebras.

Throughout this paper, we assume that the algebras are
finite-dimensional. Obvious proofs are omitted.

\section{$n$-Lie algebra and Metric $n$-Lie algebra}

\subsection{$n$-Lie algebra}
An $n$-Lie algebra is a vector space $\mathfrak g$ over a field
$\mathbb F$ equipped with an $n$-multilinear operation [$x_1,
\cdots, x_n$] satisfying
\begin{equation}[x_1, \cdots, x_n]=\text{sgn}(\sigma)[x_{\sigma(1)}, \cdots, x_{\sigma(n)}], \end{equation}
\begin{equation}\label{def1}
[x_1, \cdots, x_{n-1}, [y_1, \cdots, y_n]]=\sum_{i=1}^n [y_1,
\cdots, [x_1,\cdots, x_{n-1}, y_i], \cdots, y_n]
\end{equation}
for all $x_1, \cdots, x_{n-1}, y_1, \cdots, y_n \in \mathfrak g$,
and $\sigma \in S_n$. Identity (\ref{def1}) is usually called the
generalized Jacobi identity, or simply the Jacobi identity.

A subspace $B$ of an $n$-Lie algebra ${\mathfrak g}$ is called a
subalgebra if $[x_1, \cdots, x_n]\in B$ for all $x_1, \cdots, x_n
\in B$; whereas an ideal $I$ is a subspace of ${\mathfrak g}$ if
$[x_1, x_2 \cdots, x_n]\in B$ for all $x_1\in I, x_2, \cdots, x_n
\in {\mathfrak g}$. Let $A_1, \cdots, A_n$ be subalgebras of an
$n$-Lie algebra ${\mathfrak g}$. Denote by $[A_1, A_2, \cdots, A_n]$
the subspace of ${\mathfrak g}$ generated by all vectors $[x_1, x_2,
\cdots, x_n]$, where $x_i \in A_i$ for $i=1, 2, \cdots, n.$ The
subalgebra ${\mathfrak g}^{1}=[{\mathfrak g}, {\mathfrak g}, \cdots,
{\mathfrak g}]$ is called the derived algebra of ${\mathfrak g}$. If
${\mathfrak g}^{1}=0$, then ${\mathfrak g}$ is called an abelian
$n$-Lie algebra. The subset $C({\mathfrak g})=\{x\in {\mathfrak
g}|[x, y_1, \cdots, y_{n-1}]=0, \forall y_1, \cdots, y_{n-1} \in
{\mathfrak g}\}$ is called the center of $\mathfrak g$. An ideal $I$
is said to be maximal if any other ideal $J$ containing $I$ is
either $\mathfrak g$ and $I$. An $n$-Lie algebra is said to be
simple if it has no proper ideals and dim$\mathfrak g^1>0$.

\begin{lemma}If $I$ is a maximal ideal, then ${\mathfrak g}/I$ is simple or
one-dimensional.
\end{lemma}

The classification of simple real $n$-Lie algebras is given as
follows.

\begin{theorem}[\cite{Li}]\label{T1}
A simple real $n$-Lie algebra is isomorphic to one of the $(n+1)$-
dimensional $n$-Lie algebras defined, relative to a basis $e_i$, by
\begin{equation}\label{6}
[e_1,\cdots,\hat{e_i},\cdots,e_{n+1}]=(-1)^{i}\varepsilon_ie_i,
1\leqslant i \leqslant n+1
\end{equation}
where symbol $\hat{e_i}$ means that $e_i$ is omitted in the bracket
the $\varepsilon_i$ are signs.
\end{theorem}

It is plain to see that simple real $n$-Lie algebras admit invariant
inner products \cite{J}. Indeed, the $n$-Lie algebra in Theorem
\ref{T1} leaves invariant the diagonal inner product with entries
$(\varepsilon_1,\cdots,\varepsilon_{n+1})$.

An ideal $I$ of an $n$-Lie algebra $\mathfrak g$ is called an
$n$-solvable ideal if $I^{(r,n)}=0$ for some $r\geqslant 0$, where
$I^{(0,n)}=I$ and $I^{(s+1,n)}=[I^{(s,n)},\cdots, I^{(s,n)}]$ for
$s\geqslant 0$ by induction. The maximal $n$-solvable ideal of
$\mathfrak g$ is called the $n$-radical of $\mathfrak g$. An $n$-Lie
algebra $\mathfrak g$ is said to be strong semisimple if the
$n$-radical of $\mathfrak g$ is zero.

\begin{lemma}[\cite{Li}]\label{l1}
Let $\mathfrak g$ be a finite-dimensional $n$-Lie algebra over an
algebraically closed field of characteristic zero. Then $\mathfrak
g$ has a decomposition $\mathfrak g = \mathfrak s \oplus \mathfrak
r$, where $\mathfrak s$ is a strong semisimple subalgebra and
$\mathfrak r$ is the $n$-radical of $\mathfrak g$.
\end{lemma}

\begin{lemma}[\cite{BM}]\label{l2}
Let $\mathfrak g$ be a finite-dimensional $n$-Lie algebra over an
algebraically closed field of characteristic zero. Then $\mathfrak
g$ is strong semisimple if and only if $\mathfrak g$ can be
decomposed into the direct sum of its simple ideals.
\end{lemma}

An $n$-Lie module over an $n$-Lie algebra $\mathfrak g$ is defined
as a vector space $M$ such that the semi-direct sum $\mathfrak g +
M$ satisfies the generalized Jacobi identity. This means that the
brackets $[, \cdots , ]$ on $\mathfrak g$ is continued to $\mathfrak
g+ M$ such that $[x_1,\cdots, x_n]=0$, if at least two arguments
among $x_1, \cdots, x_n$ belong to $M$ and the generalized Jacobi
identity is true for all $x_1,\cdots, x_n \in \mathfrak g+M$. In
other words,
\begin{align*}
& [x_1,\cdots, x_{i-1},m,x_{i+1},\cdots, x_n]=-[x_1, \cdots,
x_{i-1},x_{i+1},
m, \cdots, x_n],\quad 1\leqslant i< n, \\
& [x_1, \cdots, x_{n-1}, [x_n, \cdots, x_{2n-2}, m]]-[x_n, \cdots,
x_{2n-2}, [x_1, \cdots, x_{n-1}, m]]  \\
& =\sum_{i=n}^{2n-2} [x_n, \cdots, x_{i-1}, [x_1, \cdots, x_{n-1},
x_i], x_{i+1}, \cdots, x_{2n-2}, m], \\
& [x_1, \cdots, x_{n-2}, [x_n, \cdots, x_{2n-1}],m] \\
& =\sum_{i=n}^{2n-1} [x_n, \cdots, x_{i-1}, [x_1, \cdots, x_{n-2},
x_i, m], x_{i+1}, \cdots, x_{2n-1}],
\end{align*}
for all $x_1, \cdots, x_{n-2}, x_n, \cdots, x_{2n-1} \in \mathfrak
g$ and $m \in M$.

A subspace $N$ of $M$ is called an $n$-Lie submodule of $M$ if
$[x_1, \cdots, x_{i-1}, m, x_{i+1}, \cdots, x_n] \in N$, for all
$m\in N$ and $x_1, \cdots, \hat{x_i}, \cdots, x_n \in \mathfrak g$.
$M$ is said to be irreducible if any submodule is trivial and
completely reducible if it can be decomposed to a direct sum of
irreducible submodules.

\subsection{Metric $n$-Lie algebra}
A metric $n$-Lie algebra $\mathfrak g$ is an $n$-Lie algebra over
$\mathbb{R}$ with a symmetric non-degenerate bilinear form $\langle
, \rangle $ satisfying
\begin{equation}\label{def2}\langle [x_1, \cdots, x_n], y \rangle + \langle x_n, [x_1, \cdots, x_{n-1}, y] \rangle = 0 \end{equation}
for all $x_1, \cdots, x_n , y \in \mathfrak g$.

Given two metric $n$-Lie algebras $\mathfrak g_1$ and $\mathfrak
g_2$, we may form their orthogonal direct sum $\mathfrak g_1 \oplus
\mathfrak g_2$, by declaring that
\begin{equation*}
 [x_1, x_2, y_1,\cdots,y_{n-2}]=0 \text{ and } \langle x_1, x_2\rangle=0,
\end{equation*}
for all $x_i\in {\mathfrak g_1}$ and $y_i\in \mathfrak g_1 \oplus
\mathfrak g_2$. The resulting object is again a metric $n$-Lie
algebra. A metric $n$-Lie algebra is said to be indecomposable if it
is not isomorphic to an orthogonal direct sum of metric $n$-Lie
algebras $\mathfrak g_1 \oplus \mathfrak g_2$ with dim$\mathfrak
g_i>0$.

If $W \subseteq {\mathfrak g}$ is any subspace, we define
\begin{equation*}
W^\bot=\{v\in {\mathfrak g} |\langle v,w\rangle=0, \forall w \in
W\}.
\end{equation*}
Notice that $(W^\bot)^\bot=W$. We say that $W$ is non-degenerate, if
$W \cap W^\bot=0$, whence $V=W\oplus W^\bot$; isotropic, if
$W\subseteq W^\bot$; Of course, in positive definite signature, all
subspaces are non-degenerate.

An equivalent criterion for decomposability is the existence of a
proper non-degenerate ideal: for if $I$ is a non-degenerate ideal,
${\mathfrak g}=I\oplus I^\bot$ is an orthogonal direct sum of
ideals.

\begin{lemma}[\cite{J}]\label{2}
If $\mathfrak g$ is a metric n-Lie algebra, then $C(\mathfrak
g)=(\mathfrak g^{1})^{\bot}$.
\end{lemma}

\section{Some properties on $(n+3)$-dimensional $n$-Lie algebras}
In this section, we firstly get some properties on
$(n+3)$-dimensional $n$-Lie algebras, which play a crucial role in
the classification of $(n+3)$-dimensional metric $n$-Lie algebras.
Let $V_n$ be the unique simple $n$-Lie algebra over the complex
numbers field, defined by
$$[e_1,\cdots,\hat{e_i},\cdots,e_{n+1}]=(-1)^ie_i,$$ where
$e_1,\cdots,e_{n+1}$ is a basis of $V_n$. Then we have:

\begin{lemma}[\cite{DAS}]\label{comre}
Any finite-dimensional  $n$-Lie module of $V_n$ is completely
reducible.
\end{lemma}

\begin{lemma}[\cite{DAS}]\label{weight}
The dimension of any irreducible $n$-Lie module of $V_n$ with
highest weight $t$ is equal to
$\frac{n+2t-1}{n+t-1}\binom{n+t-1}{t}$.
\end{lemma}

\begin{coro}\label{coro} Let $V$ be an $n$-Lie module of $V_n$ in dimension 2. Then $V=V_1\oplus V_2$, where $V_1,V_2$ are 1-dimensional $n$-Lie modules of $V_n$.
\end{coro}
\begin{proof} Assume that $V$ is an irreducible $n$-Lie module of $V_n$ with highest weight $t$. By Lemma~\ref{weight}, $t\not=0$. But for $t>0$ and $n\geqslant 3$,
\begin{equation*}
\dim V=\frac{n+2t-1}{n+t-1}\binom{n+t-1}{t}=\frac{n+2t-1}{t}\binom{n+t-2}{t-1}>2.
\end{equation*}
It contradicts $\dim V=2$. Then $V$ is reducible. By Lemma~\ref{comre}, we have $V=V_1\oplus V_2$.
\end{proof}

\begin{theorem}\label{myt}
Let $\mathfrak g$ be an $(n+3)$-dimensional $n$-Lie algebra over the
complex numbers field. Then $\dim \mathfrak g^1 \leqslant n + 2$.
\end{theorem}

\begin{proof}
If $\mathfrak g$ is $n$-solvable, then the conclusion is obvious.
Assume that $\mathfrak g$ is not $n$-solvable. Then $\mathfrak g =
V_n \oplus \mathfrak r$ by Lemma \ref{l1} and Lemma \ref{l2}, where
$\mathfrak r$ is the $n$-radical of $\mathfrak g$ in dimension 2.

Clearly, $\mathfrak r$ is an $n$-Lie module of $V_n$. By Corollary
\ref{coro}, $\mathfrak r = V_1 \oplus V_2 $, where $V_1,V_2$ are the
1-dimensional $n$-Lie module of $V_n$ with highest weight $0$. It
follows that there is a basis of $e_1, \cdots, e_{n+3}$ of
$\mathfrak g$, such that $e_1, \cdots, e_{n+1} \in V_n$, $e_{n+2}
\in V_1$, $e_{n+3} \in V_2$, and
\begin{equation}\label{zero}
\begin{split}
&[e_1,  \cdots,  \hat{e_i},  \cdots  e_{n+1}]=(-1)^{i} e_i, \quad 1\leqslant i \leqslant n+1, \\
&[e_1, \cdots, \hat{e_i}, \cdots, \hat{e_j}, \cdots, e_{n+1}, e_{n+2}]=0, \quad 1\leqslant i < j\leqslant n+1, \\
&[e_1, \cdots, \hat{e_i}, \cdots, \hat{e_j}, \cdots, e_{n+1}, e_{n+3}]=0, \quad 1\leqslant i < j\leqslant n+1.
\end{split}
\end{equation}

By Jacobi identity, the equations (\ref{zero}) and the fact that
$\mathfrak r$ is an ideal of $\mathfrak g$, we have that
\begin{equation*}
\begin{split}
 [e_1, \cdots, e_{n-2}, e_{n+2}, e_{n+3}] = &  -[[e_2, \cdots, e_{n+1}], e_2, \cdots, e_{n-2}, e_{n+2}, e_{n+3}] \\
= &  -[e_2, \cdots, e_n, [e_{n+1}, e_2, \cdots, e_{n-2}, e_{n+2}, e_{n+3}]] \\
&+\sum_{i=2}^{n-2} [e_{n+1},\cdots,[e_2,\cdots,e_n,e_i],\cdots,e_{n-2},e_{n+2},e_{n+3}] \\
& +  [e_{n+1}, e_2, \cdots, e_{n-2}, [e_2, \cdots, e_n, e_{n+2}], e_{n+3}]\\
& + [e_{n+1}, e_2, \cdots, e_{n-2}, e_{n+2}, [e_2, \cdots, e_n, e_{n+3}]] \\
 =&\ 0.
\end{split}
\end{equation*}

Similarly, $[e_{i_1}, \cdots, e_{i_{n-2}}, e_{n+2}, e_{n+3}]=0$ for
all  $1< i_1< \cdots < i_{n-2}<{n+1}$. That is, $\mathfrak g^1
=V_n$. Namely $\dim \mathfrak g^1=n+1 $.
\end{proof}

\section{classification of $(n+3)$-dimensional metric $n$-Lie
algebra}\label{s3}

By Lemma \ref{2} and Theorem \ref{myt}, it is easy to know the
center of an $(n+3)$-dimensional metric $n$-Lie algebra is nonzero.
Moreover,

\begin{lemma}\label{l3}
The center of a non-abelian $(n+3)$-dimensional metric $n$-Lie
algebra is 1- or 2- dimensional.
\end{lemma}
\begin{proof}
Clearly, $\dim C({\mathfrak g}) \geqslant 1$. If $\dim C({\mathfrak
g}) \geqslant 3$, then $\dim {\mathfrak g}^1\leqslant 1$. It
contradicts Lemma \ref{2}.
\end{proof}

Now let $\mathfrak g$ be an $(n+3)$-dimensional metric $n$-Lie
algebra. If $C(\mathfrak g)$ is non-isotropic, then there exists
$x\in C(\mathfrak g)$ such that $\langle x, x \rangle \neq 0$. Then
the decomposition ${\mathfrak g}={\mathbb R}x\oplus x^\bot$ is a
direct sum of non-degenerate ideals. In general, there is a
well-known fact:

\begin{lemma}
Let ${\mathfrak g}$ be a metric $n$-Lie algebra with a non-isotropic
center. Then $\mathfrak g$ is a direct sum ${\mathfrak g}={\mathfrak
g}_a\oplus {\mathfrak g}_i$, where ${\mathfrak g}_a$ is an abelian
ideal with non-degenerate restriction $\langle, \rangle$ on it and
${\mathfrak g}_i$ is an ideal with an isotropic center $C({\mathfrak
g}_i)$.
\end{lemma}

\begin{lemma}\label{l4} Let $\mathfrak g$ be an $(n+3)$-dimensional metric
$n$-Lie algebra with an isotropic center. Then $\dim C({\mathfrak
g})=2$.
\end{lemma}
\begin{proof} Since $\mathfrak g$ is a metric
$n$-Lie algebra with an isotropic center, we know $\mathfrak g$ is
non-abelian. By Lemma \ref{l3}, $\dim C(\mathfrak g)=1$ or $2$.
Assume that $C(\mathfrak g)$ is isotropic and $\dim C(\mathfrak
g)=1$. Let $e_1, e_2, \cdots, e_{n+3}$ be a basis of $\mathfrak g$
such that $e_1 \in C(\mathfrak g)$, $e_1, \cdots, e_{n+2} \in
{\mathfrak g}^1$, and
$$\langle e_1, e_{n+3} \rangle = 1, \quad\langle
e_i, e_i \rangle = \epsilon_i, \quad\forall 1< i< n+3 ,$$ where
$\epsilon_i$ are signs. Let $c_{i,j}$ denote $[e_2, \cdots,
\hat{e_i}, \cdots, \hat{e_j}, \cdots, e_{n+3}]$ for all $ 2\leqslant
i < j \leqslant n+3$. The nonzero brackets are only the following:
\begin{equation*}
  \begin{split}
   & c_{i,n+3}=[e_2, \cdots, \hat{e_i},  \cdots, e_{n+2}]=a_{i,n+3}^1 e_1 +
a_{i,n+3}^i e_i, \quad\forall 2\leqslant i \leqslant n+2, \\
   & c_{i,j}=[e_2, \cdots, \hat{e_i}, \cdots, \hat{e_j}, \cdots,
e_{n+3}]=a_{i,j}^i e_i + a_{i,j}^j e_j, \quad\forall  2\leqslant i <
j < n+3.
  \end{split}
\end{equation*}
By the definition, we have
\begin{equation*}
\begin{split}
\langle [e_2, \cdots, \hat{e_i},  \cdots, e_{n+2}] , e_{n+3} \rangle
+ (-1)^{n+2-j}\langle [e_2, \cdots, \hat{e_i}, \cdots, \hat{e_j},
\cdots,
e_{n+3}], e_j \rangle = 0, \\
\langle [e_2, \cdots, \hat{e_j},  \cdots, e_{n+2}] , e_{n+3} \rangle
+ (-1)^{n+1-i}\langle [e_2, \cdots, \hat{e_i}, \cdots, \hat{e_j},
\cdots, e_{n+3}], e_i \rangle = 0.
\end{split}
\end{equation*}
It follows that \begin{equation}\label{e1} a_{i,n+3}^1 +
(-1)^{n+2-j}\epsilon_j a_{i,j}^j=0, \quad a_{j,n+3}^1 +
(-1)^{n+1-i}\epsilon_i a_{i,j}^i=0.
\end{equation}
Similarly, we have
\begin{equation*}
\begin{split}
\langle [e_2, \cdots, \hat{e_i},  \cdots e_{n+2}] , e_i \rangle
   &=(-1)^{n+2-j} \langle [e_2, \cdots, \hat{e_i}, \cdots, \hat{e_j},
   \cdots,
e_{n+2}, e_j], e_i \rangle \\
   &=(-1)^{n+1-j} \langle [e_2, \cdots, \hat{e_i}, \cdots, \hat{e_j},
   \cdots,
e_{n+2}, e_i], e_j \rangle \\
   &=(-1)^{n+1-j} (-1)^{n+1-i}\langle [e_2, \cdots, \hat{e_j},
   \cdots,
e_{n+2} ], e_j \rangle \\
   &=(-1)^{i+j}\langle [e_2, \cdots, \hat{e_j}, \cdots,
e_{n+2} ], e_j \rangle.
\end{split}
\end{equation*}
That is, \begin{equation}\label{e2} (-1)^i \epsilon_i a_{i,n+3}^i =
(-1)^j \epsilon_j a_{j,n+3}^j
\end{equation}
By equations (\ref{e1}) and (\ref{e2}), we have, for all
$1<k<i<j<n+3$,
\begin{equation}\label{e3}
\begin{split}
  & a_{k,i}^k c_{k,j}-a_{k,j}^k c_{k,i} \\
=&a_{k,i}^k a_{k,j}^j e_j -a_{k,j}^k a_{k,i}^i e_i \\
=&(-1)^{n-k} \epsilon_k a_{i,n+3}^1 (-1)^{n+1-j} \epsilon_j
a_{k,n+3}^1 e_j -(-1)^{n-k} \epsilon_k a_{j,n+3}^1 (-1)^{n+1-i}
\epsilon_i a_{k,n+3}^1 e_i \\
=&(-1)^{n-k}\epsilon_k a_{i,j}^j a_{k,n+3}^1 e_j +
(-1)^{n-k}\epsilon_k a_{i,j}^i a_{k,n+3}^1 e_i\\
=&(-1)^{n-k}\epsilon_k a_{k,n+3}^1 c_{i,j}.
\end{split}
\end{equation}
Similarly, for all $1<i<j<n+3$,
\begin{equation}\label{e4}
(-1)^{n+1-i} \epsilon_i a_{i,n+3}^i c_{i,j} = -a_{j,n+3}^1
c_{i,n+3}+a_{i,n+3}^1 c_{j,n+3}.
\end{equation}
By induction, $\dim {\mathfrak g}^1 \leqslant n+1$ by equations
(\ref{e3}) and (\ref{e4}). It contradicts $\dim {\mathfrak g}^1 =
n+2$ because $\dim {\mathfrak g}^1 + \dim C({\mathfrak g}) = n+3$.
So the lemma follows.
\end{proof}

\begin{theorem}\label{thm}
Let $\mathfrak g$ be a non-abelian $(n+3)$-dimensional metric
$n$-Lie algebra. Then $\mathfrak g$ must be one of the following
cases (only non-zero products or brackets are given).
\begin{enumerate}
\item There exists a basis $x, e_1,\cdots,e_{n+2}$ of
      $\mathfrak g$ such that $\langle x,x\rangle =\pm 1$, $\langle
      e_1,e_1\rangle =\pm 1$, $\langle e_i, e_i \rangle =1$ for all
      $2\leqslant i\leqslant m$, $\langle e_i, e_i \rangle=-1$ for all $m
      < i \leqslant n+2$ and the brackets are given by
      $$[e_2, \cdots, \hat{e_r}, \cdots, e_{n+2}]=a_r e_r,$$
      where
      $2\leqslant r\leqslant n+2$ and $(-1)^i a_i \langle e_i, e_i \rangle  = (-1)^j a_j \langle
      e_j, e_j \rangle $ for all $2\leqslant i, j \leqslant n+2$.
\item There exists a basis $x, e_1,\cdots,e_{n+2}$ of
      $\mathfrak g$ such that $\langle x,x\rangle =\pm 1$, $\langle e_1,
      e_{n+2}\rangle = 1$, $\langle e_i, e_i \rangle =1$ for all
      $2\leqslant i\leqslant m$, $\langle e_i, e_i \rangle=-1$ for all $m
      < i \leqslant n+1$ and the brackets are given by
      \begin{equation*}
      [e_2, \cdots, e_{n+1}]= a_1 e_1, \quad [e_2, \cdots, \hat{e_i},
      \cdots, e_{n+2}]= a_i e_i,
      \end{equation*}
      where
      $2\leqslant i\leqslant n+1$ and $(-1)^n a_1 = (-1)^i a_i \langle e_i, e_i \rangle .$
\item There exists a basis $e_1,\cdots,e_{n+3}$ of $\mathfrak g$
      such that $\langle e_1, e_{n+3}\rangle =1$, $\langle
      e_2,e_{n+2}\rangle =1$, $\langle e_i, e_i \rangle =1$ for all
      $3\leqslant i\leqslant m$, $\langle e_i, e_i \rangle=-1$ for all $m
      < i \leqslant n+1$ and the brackets are given by
      \begin{equation*}
      \begin{split}
      & [e_3, \cdots, e_{n+1}, e_{n+2}]=a_1 e_1, \quad [e_3, \cdots, e_{n+1}, e_{n+3}]=a_2 e_2, \\
      & [e_3, \cdots, \hat{e_i},  \cdots, e_{n+1}, e_{n+2},   e_{n+3}]=a_i e_i,
\end{split}
\end{equation*}
where $3\leqslant i\leqslant n+1$ and $(-1)^{n+1} a_1 =
(-1)^{n+2}a_2 =(-1)^i a_i \langle e_i, e_i \rangle .$
\end{enumerate}
\end{theorem}
\begin{proof} If $C(\mathfrak g)$ is non-isotropic, then we have the cases (1) and (2) by the classification of
$(n+2)$-dimensional metric $n$-Lie algebras \cite{RCL}. If
$C(\mathfrak g)$ is isotropic, then $\dim C(\mathfrak g)=2$ by Lemma
\ref{l4}. Therefore, there is a basis of $e_1, \cdots, e_{n+3}$ of
$\mathfrak g$ such that $e_1, e_2 \in C(\mathfrak g), e_1, \cdots,
e_{n+1} \in \mathfrak g^1$ and
$$\langle e_1, e_{n+3} \rangle = \langle e_2, e_{n+2} \rangle =1,
\quad \langle e_i, e_i \rangle = \epsilon_i,$$ where $ 3\leqslant
i\leqslant n+1$. Then we have
\begin{equation}\label{id1}
\begin{split}
&[e_3, \cdots, e_{n+2}]=a_1 e_1,\quad [e_3, \cdots, e_{n+1}, e_{n+3}]=a_2 e_2, \\
&[e_3, \cdots, \hat{e_i}, \cdots, e_{n+2}, e_{n+3}]=a_i e_i, \quad
3\leqslant i\leqslant n+1.
\end{split}
\end{equation}
Since $\dim \mathfrak g^1=n+1$, we know $a_i \neq 0$ for all
$1\leqslant i\leqslant n+1$. Also we have
\begin{equation*}
\begin{split}
 \langle [e_3, \cdots, \hat{e_i}, \cdots, e_{n+2}, e_{n+3}], e_i \rangle
= & - \langle [e_3, \cdots, \hat{e_i}, \cdots, e_{n+2}, e_i],
e_{n+3}
\rangle \\
= & - (-1)^{n+2-i} \langle [e_3, \cdots, e_{n+2}], e_{n+3} \rangle.
\end{split}
\end{equation*}
It follows that $(-1)^{n+1}a_1=(-1)^i a_i \epsilon_i$. Similarly,
$(-1)^{n+2}a_2=(-1)^i a_i \epsilon_i$. It is a direct calculation to
check that (\ref{id1}) satisfies the Jacobi identity. Then the
theorem follows. \end{proof}

\begin{coro}
If $\mathfrak g$ is a non-abelian $(n+3)$-dimensional metric $n$-Lie
algebra, then $\dim C(\mathfrak g) = 2$ and $\dim {\mathfrak g}^1 =
n+1.$
\end{coro}

\section{Acknowledgements}
We would like to express our thanks to Prof. Z.X. Zhang for the
communication in this field.

\end{document}